\documentclass[letterpaper]{article} 
\usepackage{aaai2026}  
\usepackage{times}  
\usepackage{helvet}  
\usepackage{courier}  
\usepackage[hyphens]{url}  
\usepackage{graphicx} 
\usepackage{natbib}  
\usepackage{caption} 
\usepackage{algorithm}
\usepackage{algorithmic}

\usepackage{xcolor}

\usepackage{newfloat}
\usepackage{listings}
\DeclareCaptionStyle{ruled}{labelfont=normalfont,labelsep=colon,strut=off} 
\floatstyle{ruled}
\newfloat{listing}{tb}{lst}{}
\floatname{listing}{Listing}
\title{Rewiring the Majority: Structural Median Opinion Control\\ in the Friedkin–Johnsen Model}

\author {
    Elisabetta Biondi,
    Chiara Boldrini,
    Andrea Passarella,
    Marco Conti
}
\affiliations {
    IIT-CNR, Italy\\
    first.last@iit.cnr.it
}

\usepackage{amsmath}
\usepackage{amsthm}
\usepackage{amsfonts}
\usepackage{booktabs}

\DeclareMathOperator*{\argmin}{\mathrm{argmin}}
\DeclareMathOperator*{\argmax}{\mathrm{argmax}}

\begin{document}
\nocopyright

\newtheorem{problem}{Problem}
\newtheorem{theorem}{Theorem}
\newtheorem{corollary}{Corollary}

\newcommand*{\etaltwo}{\textit{et al}.\@ }

\maketitle

\begin{abstract}
Opinion control under the Friedkin-Johnsen (FJ) model has traditionally focused on modifying individual users, for example by changing their opinions or susceptibility to social influence. We study a complementary problem: can collective outcomes be steered by modifying the interaction network instead? We formulate structural median control as an edge-addition  problem aimed at moving the equilibrium median beyond a target threshold. We show that the problem is NP-hard even under highly restricted settings and propose both continuous optimization methods and computationally lighter opinion-aware heuristics for selecting edge additions. Experiments on synthetic and real-world networks reveal that intervention success depends jointly on network topology, the placement of favorable opinions, and user susceptibility. Structural interventions produce the largest shifts when opinion placement and network structure allow added edges to create useful influence pathways toward the target outcome. Finally, under a homogeneous reference setting, we compare structural and node-level interventions using normalized budgets and a degree-weighted cost model. We show that, although node interventions often appear stronger under normalized budgets, structural interventions become competitive and frequently more cost-effective once the higher cost of targeting  high-reach  users is taken into account.
\end{abstract}


\section{Introduction}

Online social platforms mediate a large fraction of social and political interaction. By shaping exposure to information and patterns of interaction, these platforms provide the substrate on which opinions are formed, reinforced, or contested. Understanding how opinions evolve over networks, and how collective outcomes depend on the underlying interaction structure, is therefore a central problem in computational social science, network science, and web mining~\cite{proskurnikov2017tutorial,proskurnikov2018tutorial,grabisch2020survey}.

A widely used model for studying such processes is the Friedkin–Johnsen (FJ) model~\cite{friedkin1990social,friedkin1997social}. Unlike classical consensus models such as DeGroot’s model~\cite{degroot1974reaching}, the FJ model incorporates two important psychological and sociological mechanisms: susceptibility to social influence and partial anchoring to agents' innate opinions; together they allow equilibrium states with persistent disagreement. This feature has made it a standard framework for studying polarization, disagreement, and opinion formation in social networks~\cite{ Matakos2017,biondi2023dynamics}. Building on this framework, a growing body of work has investigated opinion control problems, where the goal is to identify limited interventions that steer the equilibrium of the system toward a desired outcome. Existing approaches typically intervene at the node level, for instance by fixing selected agents’ expressed opinions~\cite{gionis2013opinion}, modifying innate opinions~\cite{Tu2023, sun2023opinion, gaitonde2020adversarial,chen2022adversarial}, or changing agents’ susceptibility to social influence~\cite{abebe2022opinion,chan2021hardness}.

Recent work has also considered objectives that go beyond aggregate opinion measures. 
In particular, \citet{ristache2025countering} introduced the problem of flipping the median equilibrium opinion under FJ dynamics. Unlike aggregate objectives such as the mean, the median directly captures which side of a threshold contains a majority of agents: shifting the average may leave most agents’ positions unchanged, whereas moving the median beyond a decision threshold implies changing the majority outcome. Moreover, because it depends on the ordering rather than the magnitude of opinions, the median is insensitive to a small number of extreme values, making it especially relevant in settings such as elections, referenda, and other collective decisions.  In the same work, the authors study node-level interventions on agents’ susceptibility parameters and show that the resulting median-optimization problem is computationally hard. Their work highlights the importance of median-based objectives, but leaves open a complementary question: can majority outcomes be influenced without directly modifying individual agents?

In this paper, we study this question through the lens of structural interventions, motivated by the observation that modern online platforms can modify interaction patterns more easily than users' beliefs. Through link recommendations, follow suggestions, and similar mechanisms, platforms can create new influence pathways without directly manipulating individuals~\cite{zhu2022nearly,racz2023towards}. Given a network, innate opinions, susceptibility parameters, and a limited budget, we ask whether a small number of edge additions can move the equilibrium median beyond a target threshold. Our key intuition is to create influence pathways from agents whose equilibrium opinions need to be shifted toward agents already aligned with the target outcome. We instantiate this idea through continuous relaxations that optimize smooth approximations of the median objective, together with computationally lighter opinion-aware heuristics that prioritize promising edges using the current equilibrium opinion landscape. 

Our work therefore bridges median-based opinion control and structural network interventions. This combination leads to a distinct optimization problem: the objective is non-smooth because it depends on the median, while the decision variables affect the FJ equilibrium nonlinearly through the influence matrix. Structural median control therefore requires new formulations as well as a dedicated empirical comparison with node-level manipulation.
Our study focuses on three questions: (RQ1) Under which combinations of topology, opinion placement, and susceptibility can structural interventions flip the median?  (RQ2) How effective are structural interventions on real-world networks with empirically derived opinion distributions? (RQ3) How do they compare to node-level interventions when accounting for normalized and cost-calibrated budgets? 


We evaluate homogeneous and heterogeneous susceptibility regimes on synthetic and real-world networks. Intervention feasibility depends on the alignment of network topology, opinions, and susceptibility, with heterogeneity either hindering or facilitating median control. Under homogeneous susceptibility, structural interventions are competitive with node manipulation and become increasingly cost-effective as targeting high-reach users grows more expensive. On real-world networks, selected edge additions reverse the equilibrium majority with budgets of at most 27\% of the original edge count; notably, on Polblogs, a structural method succeeds with approximately 11\% edge additions while no tested node-level method reaches the target.


Our contributions are fourfold. First, we introduce structural median control under FJ dynamics and show that its edge-addition decision problem remains NP-hard under restricted settings. Second, we develop continuous optimization methods and computationally lighter opinion-aware heuristics for selecting added edges. Third, through repeated controlled configurations and real-world opinion data, we characterize how intervention success depends jointly on opinion placement, topology, community structure, and susceptibility. Finally, we introduce a degree-weighted cost model for comparing node and structural interventions and identify the regimes in which edge additions become more cost-effective.

\section{Related Work}

Opinion-control research under Friedkin--Johnsen dynamics has primarily acted on node-level quantities. Prior work fixes selected expressed opinions~\cite{gionis2013opinion}, perturbs innate opinions~\cite{gaitonde2020adversarial,chen2022adversarial}, optimizes interventions under limited information~\cite{Tu2023}, modifies susceptibility~\cite{abebe2022opinion,chan2021hardness}, or considers opinion optimization in directed networks~\cite{sun2023opinion}. These studies show that limited interventions can substantially alter equilibrium outcomes, but they leave the interaction structure unchanged.

The closest work is \citet{ristache2025countering}, who formulate median flipping through node-level susceptibility interventions and develop continuous and combinatorial optimization methods. We adopt the same median-centric objective but intervene on the influence network: rather than modifying individual agents, we add edges that change who influences whom.

Our work also relates to structural design in opinion dynamics. Edge additions and network perturbations have been used to reduce polarization, disagreement, conflict, or distance from consensus~\cite{musco2018minimizing,zhu2021minimizing,zhu2022nearly,racz2023towards}. In contrast, we optimize a threshold-crossing median outcome. Our edge-addition formulation is motivated by link-recommendation systems that reshape interaction networks~\cite{gupta2013wtf}. We assume that selected links are adopted; incorporating probabilistic acceptance~\cite{yang2013maximizing} is left for future work.

\section{Modeling Framework} 
In this section, we provide a concise overview of the FJ model, which outlines the dynamics of opinion evolution. Following this, we present various formulations of the primary problem addressed in our work: investigating the effectiveness of edge interventions in steering the final opinions' median toward a desired direction.

\subsection{Friedkin-Johnsen Opinion Dynamics}\label{sec:setting}

Let $G=(V,E)$ be an unweighted directed social network with $n=|V|$ nodes representing the users and $E$ representing the edges. An edge $(i,j)\in E$ indicates that user $i$ follows user $j$ and is therefore influenced by $j$, for instance by being exposed to $j$'s posts. Thus, the edge orientation represents the \emph{following} relation, while information and influence flow in the opposite direction. We denote by $\mathcal{N}^{out}(i)$ the set of users followed by $i$. Let $A=(a_{ij})$ be the binary adjacency matrix of $G$, where $a_{ij}=1$ if $(i,j)\in E$ and $a_{ij}=0$ otherwise. We assume that each user assigns equal weight to the opinions of all users they follow. Thus, $w_{ij}$ quantifies the contribution of user $j$'s opinion to user $i$'s update, and the row-stochastic influence matrix $W=(w_{ij})$ is obtained by row-normalizing $A$, i.e., $
w_{ij}=\frac{a_{ij}}{\sum_{k\in V}a_{ik}}=\frac{a_{ij}}{d^{out}(i)}$, where $d^{out}(i)=|\mathcal{N}^{out}(i)|$. Thus, $w_{ij}=1/d^{out}(i)$ if user $i$ follows user $j$, and $w_{ij}=0$ otherwise. In the FJ model, each node $i\in V$ is associated with an innate opinion $s_i \in [0,1]$ and a susceptibility parameter $\lambda_i \in [0,1]$, which represents the extent to which user $i$ moves from its initial opinion under social influence. 
Opinions evolve in discrete time according to the following rule:
$$x_i(t+1)=(1-\lambda_i)s_i+\lambda_i \sum_{j\in \mathcal{N}^{out}(i)} w_{ij} x_j(t),$$
where $\mathcal{N}^{out}(i)$ indicates the out-neighbourhood of node $i$. If $\Lambda$ is a diagonal matrix with diagonal entries equal to $\lambda_1,\lambda_2,\dots$, then the matrix formulation of the update rule can be written as $x(t+1)=(I-\Lambda)s+\Lambda W x(t)$.
Under mild conditions, the dynamics converge to the unique equilibrium $z=(I-\Lambda W)^{-1}(I-\Lambda)s$.
This happens, for instance, if $\lambda_i < 1$ for all nodes, as the spectral radius of $\Lambda W$ satisfies $\rho(\Lambda W)<1$. Interested readers can refer to \cite{proskurnikov2017tutorial} for a more in-depth discussion on the convergence properties of FJ model.


\subsection{Structural interventions}
\label{sec:structural_intervention}

We consider structural interventions that modify the network topology by adding new directed edges. Let
$$
E_{\mathrm{lat}}=\{(i,j)\in V\times V:\; i\neq j,\ (i,j)\notin E\}
$$
denote the set of latent edges, namely all missing directed edges excluding self-loops. Then, the intervention activates a subset $F\subseteq E_{\mathrm{lat}}$, producing the modified graph $G_F=(V,E\cup F)$. Consistently with our unweighted setting, influence is distributed uniformly among all outgoing neighbors. The influence matrix $W(F)$ is therefore obtained by row-normalizing the binary adjacency matrix $A(F)=(a_{ij}(F))$ of $G_F$, where:
$$a_{ij}(F)=
\begin{cases}
1 & \text{if } (i,j)\in E\cup F,\\
0 & \text{otherwise}.
\end{cases}$$
Thus
$$w_{ij}(F)=\frac{a_{ij}(F)}{\sum_{k\in V}a_{ik}(F)}.$$

Although other structural interventions, such as edge rewiring or edge removal, could also be considered, we focus exclusively on edge addition. This setting naturally models recommendation systems that increase users' exposure by creating new interactions while preserving the existing network structure. Extending the formulation to weighted networks would additionally require specifying the influence assigned to newly added edges and how existing influence weights are redistributed. We leave these extensions for future work.

\subsection{Problem Formulation}
\label{sec:problems}

Our objective is to steer the equilibrium of the FJ model through structural interventions. Since majority outcomes are naturally captured by the median equilibrium opinion, we consider optimization objectives based on the median after adding a set of latent edges. We denote by
$$
M(z(F))=\mathrm{Median}(z(F)),
$$
the equilibrium median obtained after adding an edge set $F\subseteq E_{\mathrm{lat}}$ and recomputing the row-normalized influence matrix $W(F)$. We study two optimization formulations, presented in the increasing direction without loss of generality.  Indeed, since the influence matrix $W(F)$ is row-stochastic, we have $(I-\Lambda W(F))\mathbf{1}=(I-\Lambda)\mathbf{1}$. Therefore, replacing the innate opinions with $\tilde{s}=\mathbf{1}-s$ yields
$$
\tilde{z}(F)=(I-\Lambda W(F))^{-1}(I-\Lambda)(\mathbf{1}-s)
=\mathbf{1}-z(F),
$$
and hence $M(\tilde{z}(F))=1-M(z(F))$ for every added edge set $F$. Thus, minimizing the original median is equivalent to maximizing the transformed median. Similarly, moving the original median to or below $\theta$ is equivalent to moving the transformed median to or above $1-\theta$. The admissible edge sets and their cardinalities remain unchanged.

\begin{problem}[Structural Median Maximization]
\label{pr:maximize}
Given a budget $k\in\mathbb{N}$ of the number of possible interventions, find
$$
F^* \in \argmax_{F\subseteq E_{\mathrm{lat}},\, |F|\le k}
M(z(F)).
$$
\end{problem}

\begin{problem}[Structural Median Threshold Minimization]
\label{pr:threshold_min}
Given a threshold $\theta\in[0,1]$, find
$$
F^* \in \argmin_{F\subseteq E_{\mathrm{lat}}} |F|
$$
such that
$$
M(z(F)) \ge \theta .
$$
\end{problem}

When the initial median lies below $\theta$, the threshold-minimization problem corresponds to flipping the median across the target threshold. This formulation is particularly relevant in applications such as elections or referenda, where the objective is to secure a majority rather than maximize the margin. To establish computational hardness, we consider the following decision version.

\begin{problem}[Structural Median Threshold Decision]
\label{pr:threshold_decision}
Given a budget $k$ and a threshold $\theta$, decide whether there exists a set $F\subseteq E_{\mathrm{lat}}$ with $|F|\le k$ such that
$$
M(z(F)) \ge \theta .
$$
\end{problem}
\noindent The decision formulation is introduced solely for the complexity analysis, whereas the remainder of the paper focuses on the two optimization problems.
\section{Computational Hardness}
\label{sec:hardness}

We show that structural median control is computationally intractable in general.

\begin{theorem}
\label{thm:ea-threshold-decision-hard}
The \textsc{Structural Median Threshold Decision} problem is NP-hard, even when all missing non-self edges are admissible, the innate opinions are binary, all susceptibility parameters are homogeneous, and the threshold is fixed to a constant strictly between zero and one.
\end{theorem}

\begin{proof}[Proof sketch]
We reduce from \textsc{Vertex Cover}. Given an undirected graph $H=(Q,L)$ and an integer $k$, the problem asks whether there exists a set $C\subseteq Q$, with $|C|\leq k$, containing at least one endpoint of every edge in $L$. At a high level, the reduction represents each vertex of $H$ by a two-user gadget with zero innate opinions. Selecting a vertex corresponds to adding one edge from its gadget to a fixed positive anchor, causing both users in the gadget to move above the threshold. Each edge of $H$ is represented by a target user influenced by the representatives of its endpoints. Activating either endpoint gadget therefore also moves the corresponding target above the threshold without additional edge additions, mirroring how selecting a vertex covers all its incident edges.

More precisely, for each $q_j\in Q$, we construct a directed two-user cycle $\{v_j,b_j\}$. For each edge $\ell=\{q_i,q_j\}\in L$, we introduce a target $u_\ell$ influenced by $v_i$ and $v_j$. We also introduce a cycle of $P$ innate-one users and a suitable number of zero-opinion padding pairs. The padding is chosen so that, letting $r=|L|$, the network contains $N=2(P+2k+r)-1$ users. Susceptibility parameters equal $1/2$, all missing non-self edges are admissible, and the threshold is fixed to $1/28$.

If $H$ admits a vertex cover of size at most $k$, we add one edge from each selected vertex gadget to a positive anchor. These additions make the $2k$ users in the selected gadgets and all $r$ targets exceed the threshold. Together with the $P$ positive anchors, this yields $P+2k+r=(N+1)/2$ users above the threshold, and hence moves the median above it.

Conversely, every activated vertex gadget, padding pair, or target not activated through one of its endpoint gadgets requires a distinct added outgoing edge. If no vertex cover of size at most $k$ exists, the resulting counting argument bounds the number of positive users under any intervention of size at most $k$ by $P+2k+r-1$, one below the number required to cross the median threshold.

The construction is polynomial and a successful intervention exists if and only if $H$ admits a vertex cover of size at most $k$. The complete construction and counting argument are provided in the SI.
\end{proof}

The result holds under highly simplified conditions: binary opinions, homogeneous susceptibilities, and a fixed threshold. Hence, the hardness is not due to complex opinion spaces or heterogeneous agent parameters, but to the combinatorial nature of selecting structural interventions.

\begin{corollary}
\label{cor:ea-optimization-hard}
The Edge-Addition Median Maximization Problem (Problem~\ref{pr:maximize})  and the Edge-Addition Median Threshold Minimization Problem (Problem~\ref{pr:threshold_min}) are NP-hard.
\end{corollary}

This follows immediately because the flipping problem is the threshold decision version of median maximization.

\vspace{-5pt}
\section{Structural Optimization Methods}
\label{sec:methods}

The NP-hardness result makes exact optimization impractical at scale. However, approximation remains challenging: equilibrium opinions depend nonlinearly on the influence matrix, the median is non-smooth, and the space of edge subsets is combinatorial. We propose two complementary families of structural intervention methods. The first family uses a continuous relaxation of edge additions: latent edges are assigned continuous activation variables, and smooth surrogate objectives are optimized over the resulting influence matrix. The second family consists of discrete opinion-aware heuristics that operate directly on candidate edges and exploit the current equilibrium opinion landscape. The continuous methods provide a principled optimization framework, whereas the heuristics are designed to be efficient, interpretable, and closely aligned with the median-flipping objective. 
Figure~\ref{fig:strategies} schematically illustrates the
selection principles underlying the four policies.
\begin{figure}[h]
\centering
\includegraphics[width=0.8\columnwidth]{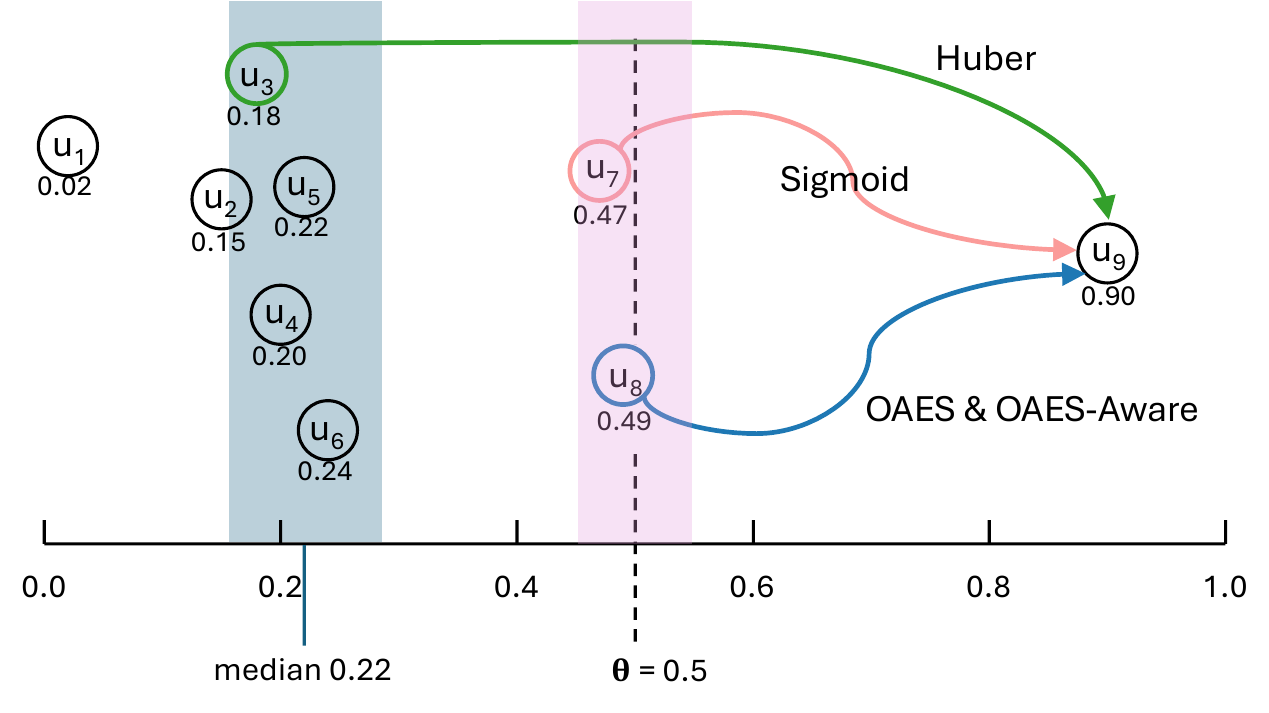}
\vspace{-1pt}
\caption{Schematic illustration of the four edge-selection policies. Nodes are positioned by their equilibrium opinions; shaded bands highlight the central opinion region and the neighborhood of $\theta$. Colored arrows show illustrative, alternative edge additions from receiver to source (\emph{following} relation), not computed selections on a specific network.}
\label{fig:strategies}
\vspace{-10pt}
\end{figure}

\subsection{Continuous Optimization}\label{sec:methods_continuous}

We begin by formulating edge addition as a continuous optimization problem. For each latent edge $(i,j)\in E_{\mathrm{lat}}$, we introduce a continuous influence-allocation variable $y_{ij}\in[0,1]$. Unlike node-level interventions, where each decision variable modifies a single agent parameter, an edge-allocation variable changes an entire row of the influence matrix. Since influence weights are obtained by row-normalizing the adjacency matrix, activating a new outgoing edge from node $i$ not only creates a new influence channel from $j$ to $i$, but also redistributes the influence mass previously assigned to $i$'s existing neighbors.

Let
$$
L_i=\{j:(i,j)\in E_{\mathrm{lat}}\}
$$
be the set of latent outgoing neighbors of node $i$. The variables $y_{ij}$ reallocate part of the outgoing influence mass of $i$ toward latent edges. To preserve row-stochasticity, we impose
$$
\sum_{j\in L_i} y_{ij}\le 1 \qquad \forall i\in V .
$$
Given the original influence matrix $W^{(0)}$, this yields the following mass-reallocation relaxation:
\begin{equation*}
 W_{i\ell}(y)=
\begin{cases}
\left(1-\sum_{j\in L_i}y_{ij}\right)W^{(0)}_{i\ell}, & (i,\ell)\in E, \\[3pt]
y_{i\ell}, & (i,\ell)\in E_{\mathrm{lat}},\\[3pt]
0, & \text{otherwise}.
\end{cases}   
\end{equation*}
The corresponding equilibrium is
$$
z(y)=(I-\Lambda W(y))^{-1}(I-\Lambda)s .
$$

We consider two policies, named after their respective surrogate
objectives. The  rationale behind the \emph{Huber} policy is to increase the median rather than the average opinion, which can be strongly affected by a few extreme values. The policy uses the Huber loss function, whose linear tails limit the influence of extreme opinions and whose quadratic behavior for small deviations
smooths the loss near the center. The resulting approximation
of the median is defined as: 
$$
\hat m_c(y)
=
\arg\min_m
\sum_{i=1}^{n} H_c(z_i(y)-m),
$$
where
$$
H_c(r)=
\begin{cases}
\frac{1}{2}r^2, & |r|\le c,\\[3pt]
c\left(|r|-\frac{1}{2}c\right), & |r|>c.
\end{cases}
$$
The policy maximizes $\hat m_c(y)$ as a surrogate for the
equilibrium median. The \emph{Sigmoid} policy instead exploits the fact that reaching a majority depends on how many nodes cross the threshold, rather than how far their opinions move beyond it. The sigmoid function is most sensitive to opinion changes near $\theta$ and saturates far from the threshold. 
Its objective is defined as:
$$
f_{\mathrm{sig}}(y)
=
\sum_{i=1}^{n}
\sigma\left(\tau(z_i(y)-\theta)\right),
\qquad
\sigma(t)=\frac{1}{1+e^{-t}} .
$$
It smoothly approximates the number of nodes whose equilibrium opinion exceeds $\theta$.

Both objectives are optimized subject to the same feasibility constraints:
$$
\sum_{(i,j)\in E_{\mathrm{lat}}} y_{ij}\le k,\qquad
\sum_{j\in L_i} y_{ij}\le 1\ \forall i,\qquad
y_{ij}\in[0,1].
$$
The relaxed problems are solved by projected gradient ascent. Although the surrogate objectives are inspired by median-control methods for node-level interventions~\cite{ristache2025countering}, the resulting optimization problem is different. In node-level control, a decision variable typically changes an agent parameter while the influence matrix remains fixed. Here, each variable enters the equilibrium through $W(y)$, changes the local normalization of an entire row, and affects the final objective through both the newly allocated edge mass and the induced reweighting of pre-existing neighbors. Moreover, this dependence is nonlinear, since changes in $W(y)$ are propagated through the inverse operator $(I-\Lambda W(y))^{-1}$ defining the FJ equilibrium.

The optimized values $y_{ij}$ are used only to rank candidate edges, and the top-$k$ edges are selected as the discrete intervention $F$. After rounding, the fractional matrix $W(y)$ is discarded: the selected edge set $F$ is added to the original adjacency matrix, and the final influence matrix $W(F)$ is obtained by row-normalizing the resulting discrete adjacency matrix, using the same rule defined for discrete structural interventions.

\subsection{Opinion-Aware Heuristics}\label{sec:methods_heuristics}

Continuous optimization provides a principled relaxation of structural median control, but it also introduces a relaxation-rounding gap. The original intervention is discrete, whereas the relaxed solution assigns fractional activation values to latent edges. Moreover, the optimized objective is a smooth surrogate of the median or of the threshold-crossing condition. As a result, a solution that is effective in the relaxed problem may not always translate into an equally effective discrete edge set, especially for small budgets or when the median changes only after specific edge combinations are selected.

To complement the continuous methods, we introduce efficient discrete heuristics that rank candidate edges using the current equilibrium opinions. Let $z$ denote the equilibrium opinion vector of the original network, and let
\begin{equation*}
b_i=(Wz)_i
\end{equation*}
be the average opinion currently observed by node $i$ through its outgoing neighborhood. For a candidate edge $(i,j)$, we define the exposure improvement
\begin{equation}
\Delta_{ij}
=
\max\{0,z_j-b_i\}.
\end{equation}
This quantity measures how much more favorable node $j$'s opinion is compared with the current neighborhood opinion of node $i$.

\subsubsection{Opinion-Aware Edge Selection (OAES)}
\label{sec:oaes}

The first heuristic prioritizes candidate edges that both provide a positive exposure improvement and affect receivers close to the target threshold. The rationale is that increasing the opinion of a node that is already far from the decision boundary may have little impact on the median, whereas shifting a near-threshold node can be decisive.

We define the receiver relevance as
\begin{equation*}
r_i
=
\exp\left(-\beta |z_i-\theta|\right),
\end{equation*}
where $\beta>0$ controls how strongly the score concentrates around the threshold. The resulting score is
\begin{equation}
\mathrm{OAES}(i,j)
=
r_i \cdot \Delta_{ij}.
\end{equation}
Thus, OAES favors edges that expose near-threshold receivers to more favorable opinion sources than those currently present in their neighborhood.

\subsubsection{Threshold-Aware Opinion-Aware Edge Selection
(OAES-Aware)}
\label{sec:oaes_aware}

The second heuristic, which we call
\emph{Threshold-Aware Opinion-Aware Edge Selection (OAES-Aware)},
further accounts for the opinion of the influencing node. Connecting a near-threshold receiver to a source whose opinion is only marginally favorable may be less effective than connecting it to a strongly favorable source. We therefore introduce a source-awareness factor
\begin{equation}
q_j
=
\frac{1}
{1+\exp\left[-\kappa(z_j-\theta)\right]},
\end{equation}
where $\kappa>0$ controls the sharpness of the transition around the threshold. The value $q_j$ is close to one for nodes well above the threshold and close to zero for nodes below it. The final score is
\begin{equation}
\operatorname{OAES\text{-}Aware}(i,j)
=
r_i \cdot q_j \cdot \Delta_{ij}.
\end{equation}
This formulation simultaneously favors receivers close to the decision boundary, sources that strongly support the target outcome, and edges that provide a substantial exposure improvement.


\subsection{Computational Complexity and Scalability}

Table~\ref{tab:complexity} reports the asymptotic complexity of the proposed methods under the dense linear-algebra implementation used in our experiments. OAES and OAES-Aware require one equilibrium computation, whereas Huber and Sigmoid require two linear-system solves per iteration. In practice, this means OAES and OAES-Aware incur roughly $\frac{1}{2T}$ of the dominant $O(n^3)$ cost required by the continuous methods, since their total cost is essentially a single equilibrium computation rather than 2T of them. Complete derivations, memory requirements, and scalability considerations are provided in the SI.

\begin{table}[h]
\centering
\caption{Computational complexity of the proposed structural intervention methods. Here $n=|V|$, $L=|E_{\mathrm{lat}}|=O(n^2)$, and $T$ is the number of projected-gradient iterations.}
\vspace{-5pt}
\label{tab:complexity}
\small
\begin{tabular}{lcc}
\toprule
Method & Linear solves (FJ) & Total time complexity \\
\midrule
OAES        & $1$  & $O(n^3+L\log L)$ \\
OAES-Aware  & $1$  & $O(n^3+L\log L)$ \\
Huber       & $2T$ & $O\!\left(T(n^3+L)+L\log L\right)$ \\
Sigmoid     & $2T$ & $O\!\left(T(n^3+L)+L\log L\right)$ \\
\bottomrule
\end{tabular}
\vspace{-10pt}
\end{table}

\section{Experimental Setup}

We evaluate the proposed structural intervention framework using controlled opinion configurations and real-world datasets with observed opinions. Controlled opinion configurations allow us to isolate the effect of opinion placement and network topology, whereas datasets with observed opinions allow us to assess whether the same mechanisms arise in real-world settings. 

Throughout the evaluation, we use $\theta=0.5$ as the decision threshold. The intervention direction is determined by the median of the innate opinions: the objective is to move the equilibrium median to the opposite side of the threshold. We refer to opinions lying on the target side of $\theta$ as \emph{favorable}; the term is purely operational and carries no
normative judgment.



\subsection{Datasets  and Opinion Configurations}
\label{sec:datasets}

We evaluate our methods on three networks with controlled opinion configurations and three real-world datasets with observed opinion scores. Table~\ref{tab:datasets} summarizes the datasets and opinion configurations used in the evaluation.

The controlled setting includes Karate Club~\cite{zachary1977information}, a Facebook friendship network~\cite{leskovec2012learning}, and a two-community Stochastic Block Model (SBM). Karate Club and Facebook provide fixed network topologies on which we generate synthetic innate opinions, while the SBM allows us to explicitly study the interaction between community structure and opinion placement. The SBM has $100$ nodes, with within- and between-community connection probabilities $p_{in}=0.3$ and $p_{out}=0.01$. For the controlled experiments, we use opinion configurations that vary both the strength and the placement of favorable opinions. In the \emph{binary} setting, favorable nodes have opinion $1$ and all other nodes have opinion $0$. In the \emph{bimodal} setting, favorable opinions are sampled uniformly from $[0.75,1]$, while unfavorable opinions are sampled uniformly from $[0,0.35]$. In \emph{central} configurations, favorable opinions are assigned to the top $20\%$ most PageRank central nodes; in \emph{random} configurations, favorable nodes are selected uniformly at random. For the SBM, we compare \emph{community-aligned} and \emph{random} placements while keeping the favorable fraction fixed at $20\%$. In the community-aligned configurations, favorable nodes are sampled within one planted community; in the random configurations, the same number of favorable nodes is sampled from the entire network.

The real-world evaluation uses three datasets with observed ideological or political scores. Reddit-Chitra and Twitter-Chitra~\cite{chitra2020filter} contain interaction networks with ideological scores inferred from user behavior, while Polblogs~\cite{adamic2005political} contains hyperlinks among political blogs labeled by political orientation. In all three cases, the observed scores are used as innate opinions in the FJ model.

\begin{table}[t]
\centering
\caption{Datasets used in the experimental evaluation.}
\vspace{-8pt}
\label{tab:datasets}
\resizebox{\columnwidth}{!}{
\begin{tabular}{lrrl}
\toprule
Dataset & $|V|$ & $|E|$ & Opinion configurations \\
\midrule
Karate Club \cite{zachary1977information} & 34 & 156 & Synthetic: central/random, binary/bimodal \\
Facebook~\cite{leskovec2012learning} & 1519 & 8078 & Synthetic: central/random, binary/bimodal \\
SBM & 100 & 1466 & Synthetic: random/community-aligned, binary/bimodal \\
\midrule
Polblogs~\cite{adamic2005political} & 1224 & 19184 & Observed political leaning \\
Reddit-Chitra~\cite{chitra2020filter} & 553 & 17938 & Observed ideological score \\
Twitter-Chitra~\cite{chitra2020filter} & 548 & 7276 & Observed ideological score \\
\bottomrule
\end{tabular}
}
\vspace{-15pt}
\end{table}

\subsection{Common Settings and Comparison Methods}
\label{sec:comparison_methods}

Unless otherwise stated, we consider both homogeneous and heterogeneous susceptibility regimes, with $\lambda_i$ measuring the extent to which node $i$ responds to social influence rather than remaining anchored to its innate opinion.  Under homogeneous susceptibility, we set $\lambda_i=0.8$ for every node.  Under heterogeneous susceptibility, we construct two policies based on PageRank. In the \emph{centrality-aligned} policy, PageRank scores are linearly rescaled to $[0.001,0.98]$, assigning higher susceptibility to more central nodes. In the \emph{centrality-complemented} policy, this association is reversed, so that more central nodes are less susceptible. For SBM, we additionally consider the analogous policies based on betweenness centrality.

For experiments targeting a decrease in the equilibrium median, we apply the opinion transformation described in the Problem Formulation section. The selected interventions are subsequently evaluated in the original opinion space.

\paragraph{Structural baselines.}
We compare against Random, Degree, PageRank, and Opinion-PR~\cite{racz2023towards,zhu2021minimizing,ristache2025countering,sun2023opinion,Tu2023,zhu2022nearly}. Candidate edge $(i,j)$ is scored through its destination $j$: Degree uses $d_j^{\mathrm{in}}$, PageRank uses $PR(j)$, and Opinion-PR uses $PR(j)s_j$. Random selects candidate edges uniformly. Deterministic policies select the highest-ranked feasible edges. 


On Karate Club, we additionally evaluate a Greedy baseline based on exact marginal evaluations. At each iteration, it selects the candidate edge producing the largest one-step improvement toward the target. Since this requires recomputing the equilibrium for every remaining candidate edge, the method is prohibitively expensive on the larger networks and is therefore used only as a small-scale reference. Its computational cost is detailed in the SI.

\paragraph{Node-level counterparts.}
We adopt the node-level intervention model and the Huber and Sigmoid optimization methods of~\citet{ristache2025countering}. Following their intervention rule, all node-selection policies act on selected nodes by making those with favorable innate opinions fully stubborn and those with unfavorable opinions maximally susceptible. We retain the Random and Degree baselines considered in their evaluation and additionally include PageRank and Opinion-PR to mirror the corresponding structural baselines.
\paragraph{Optimization and Reporting Settings.}  
Based on the sensitivity analysis reported in the SI, we use $c=0.05$, $\tau=20$, $\beta=10$, and $\kappa=25$ across all datasets. Across the tested ranges, the qualitative intervention curves remain stable; we therefore treat these values as reference settings rather than universally optimal parameters. We run the continuous methods for 40 iterations on Karate Club and 15 iterations on the larger networks. For controlled scenarios involving sampled opinion values or random node placement, we generate ten independent configurations and report the median and interquartile range across them. Random is averaged over 100 runs within each configuration before this aggregation. For fixed inputs, Random is reported through its mean and $95\%$ confidence interval over 100 runs, while all other methods are deterministic.





\subsection{Evaluation Protocols}
\label{sec:evaluation_protocols}

Since node and edge interventions modify different objects, their budgets are not directly comparable. We therefore evaluate them under two complementary protocols: a normalized intervention-fraction comparison and a degree-weighted cost sensitivity analysis.
\vspace{-2pt}
\vspace{-3pt}
\paragraph{Normalized intervention fraction.}
We first compare the two approaches using their natural intervention scale. Structural interventions are evaluated as a function of the fraction of added edges, $k_{\mathrm{edge}}/|E|$, whereas node-level interventions are evaluated using the fraction of manipulated nodes, $k_{\mathrm{node}}/|V|$. This comparison measures how the equilibrium median changes as increasingly larger portions of the corresponding intervention space are modified.

\paragraph{Degree-weighted cost sensitivity analysis.}
The normalized comparison above abstracts away from possible differences in the implementation costs of structural and node-level interventions. We therefore complement it with a stylized cost sensitivity analysis. Unit-cost edge additions correspond to the cardinality budgets commonly adopted in network-intervention and link-recommendation problems~\cite{racz2023towards,zhu2021minimizing}. For node interventions, budgeted influence-maximization models commonly assume additive and heterogeneous targeting costs
~\cite{nguyen2013budgeted}. Degree-proportional costs have also been used as a transparent proxy for the greater expense of targeting high-reach nodes~\cite{erd2021generalized}.

In our approach to edge orientation, an edge $(i,j)$ signifies that $i$ follows $j$. This means that $\deg_{\mathrm{in}}(u)$ represents the count of users who directly follow node $u$, serving as a tangible measure of its reach. Accordingly, we define the intervention cost for edge and node as follows:
\begin{equation}
C_{\mathrm{edge}}(F)=|F|,
\qquad
C_{\mathrm{node}}(S)
=
\rho\sum_{u\in S}\deg_{\mathrm{in}}(u),
\label{eq:cost_model}
\vspace{-3pt}
\end{equation}
where $\rho\geq 1$ scales the reach-weighted cost of node manipulation relative to the unit cost of adding an edge. We treat Eq.~\eqref{eq:cost_model} as a stylized sensitivity model rather than an empirically calibrated monetary cost. We therefore examine how the comparison changes as $\rho$ increases, with $\rho=1$ representing the least penalizing setting for node interventions. The main analysis reports the representative values $\rho\in\{1,2,5,10\}$, while the extended analysis up to $\rho=50$ is provided in the SI.

\section{Performance evaluation}
\label{sec:evaluation}
We evaluate our proposed structural intervention framework through a series of experiments designed to answer three research questions raised in the Introduction.

\subsection{RQ1: When can structural interventions flip the median?}
\label{sec:rq1_synthetic}

We first study structural interventions in controlled synthetic settings, which allow us to isolate two elements that are intertwined in real data: the placement of favorable opinions and the presence of community structure. In all synthetic settings, the median innate opinion lies below the target threshold $\theta=0.5$. Our objective is therefore to move the equilibrium median above this threshold, and we refer to opinions above $\theta$ as \emph{favorable}. 

\subsubsection{RQ1a: How effective are structural interventions in
networks without explicit community structure?}

Figure~\ref{fig:synthetic_noncommunity} reports the final equilibrium median obtained by structural interventions on Karate Club and Facebook under four opinion initialization schemes and homogeneous susceptibility. Curves and ribbons summarize, respectively, its median and interquartile range across ten independent opinion realizations; the central binary configuration is deterministic and evaluated once.

Random selects edges uniformly, whereas Degree and PageRank rely exclusively on network structure. Opinion-PR additionally accounts for opinion information. OAES prioritizes receivers close to the target threshold, while OAES-Aware also favors sources whose opinions are aligned with the target. Huber and Sigmoid optimize continuous surrogate objectives. On Karate Club, we additionally report a Greedy baseline based on exact marginal evaluations; its computational cost prevents its application to the larger networks.

The results show that intervention success depends strongly on the alignment between initial opinions and network topology. In the central configurations, where favorable opinions are assigned to influential nodes, several methods cross the target on Karate Club. On Facebook, 7 of the 10 central bimodal configurations already satisfy the target at zero budget; in the remaining configurations, several methods cross the threshold as the budget increases. The deterministic central binary configuration also responds strongly to intervention: OAES-Aware crosses the threshold at a small budget and several other methods succeed at larger budgets.

Random opinion placement produces a more difficult and heterogeneous intervention landscape. On Karate Club, both random configurations remain below the target, although Greedy produces the largest shifts. On Facebook, the median curves also remain below the target, but outcomes vary considerably across realizations: 3 of 10 configurations in both random settings already satisfy the target at zero budget, while others remain far below it even at large budgets. The broad interquartile ranges reflect differences in whether randomly placed favorable opinions fall on structurally influential users. Despite this variability, OAES-Aware and Opinion-PR generally provide the strongest outcomes among configurations requiring intervention. The complete heterogeneous-susceptibility results in the SI confirm that the effectiveness and relative ranking of opinion-aware methods depend on susceptibility and opinion placement.

Overall, added edges are most effective when they create influence pathways from favorable sources that are well positioned to affect pivotal nodes near the threshold. Increasing the budget often improves the equilibrium median, but does not guarantee a flip when favorable opinions are poorly aligned with the influence structure.



\begin{figure*}[h]
\centering
\includegraphics[width=0.9\textwidth]{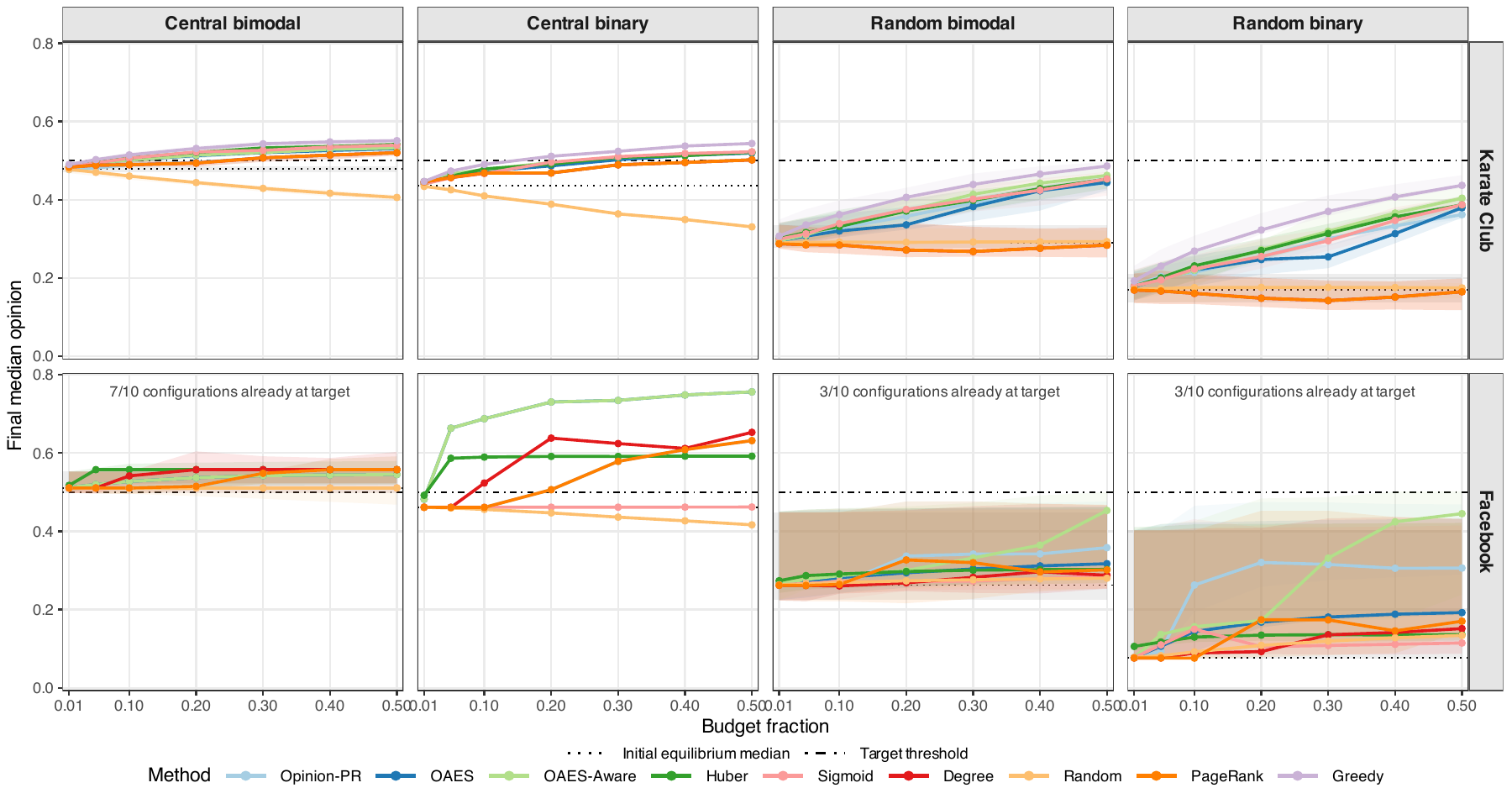}
\vspace{-8pt}
\caption{Final equilibrium median under structural interventions on Karate Club and Facebook with homogeneous susceptibility. Lines and shading show the median and interquartile range over ten opinion configurations; Random is averaged over 100 runs per configuration, while Central binary is deterministic. Dot-dashed lines mark the target; dotted lines and grey shading show the initial equilibrium. Budget is $k_{\mathrm{edge}}/|E|$.}
\label{fig:synthetic_noncommunity}
\vspace{-10pt}
\end{figure*}

\subsubsection{RQ1b: How does community structure affect structural interventions?}

We next investigate the role of community structure using the SBM network. To isolate opinion placement from opinion prevalence, all four configurations assign favorable innate opinions to exactly $20\%$ of the nodes. In the \textit{community} configurations, these nodes are concentrated within one planted community; in the \textit{random} configurations, the same number of favorable nodes is distributed across the network. Using the same median-and-interquartile-range reporting protocol, Figure~\ref{fig:synthetic_sbm} compares community-concentrated and randomly distributed favorable opinions on the SBM network under homogeneous susceptibility.

Holding the favorable share fixed leads to a more nuanced picture than the previous design. None of the configurations satisfies the target at zero budget. The \textit{community bimodal} configuration is the most amenable to intervention: OAES-Aware reaches the target in the median curve at the maximum tested budget, while OAES approaches it closely. The \textit{community binary} configuration remains below the target, although Sigmoid and the opinion-aware policies produce substantial shifts. Neither random configuration reaches the threshold within the tested budget range but, across configurations, methods incorporating opinion information consistently produce larger shifts than purely structural baselines.

The comparison reveals two effects of community concentration. Despite having the same fraction and values of favorable innate opinions, the community configurations start from lower equilibrium medians than their random counterparts, as favorable influence remains largely confined within one community. Once edges are added, however, the community configurations generally exhibit larger shifts relative to their initial equilibria. This pattern is consistent with selected inter-community edges exposing nodes outside the favorable community to a coherent region of favorable opinions, thereby creating greater intervention leverage than when favorable sources are dispersed across the network. Community structure can therefore amplify the effect of edge additions while simultaneously producing a less favorable initial equilibrium; whether the target is ultimately reached depends on the balance between these two effects.\\

\begin{figure*}[h] \centering \includegraphics[width=0.9\textwidth]{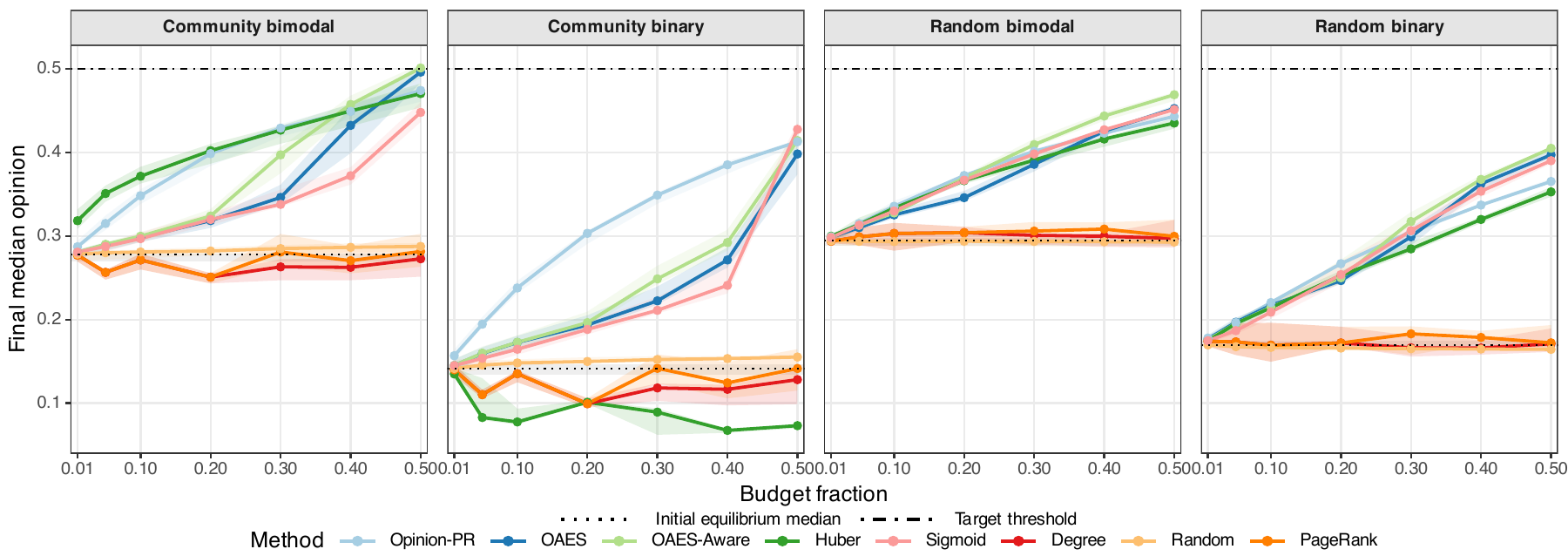} \vspace{-8pt} 
\caption{Final equilibrium median on SBM with homogeneous susceptibility and $20\%$ favorable nodes. Lines and shading show the median and interquartile range over ten opinion realizations; Random is averaged over 100 runs per realization. Dot-dashed lines mark the target; dotted lines and grey shading show the initial equilibrium. Budget is $k_{\mathrm{edge}}/|E|$.}
\label{fig:synthetic_sbm} 
\vspace{-10pt}
\end{figure*}

For lack of space, the complete results under heterogeneous susceptibility are reported in the SI. They show that the feasibility of structural intervention depends on how opinion placement and susceptibility jointly align with network topology. 

\subsection{RQ2: How effective are structural interventions on real-world networks with empirically derived opinion distributions?}
\label{sec:rq2_real}

We now evaluate structural interventions on real-world networks with empirically derived opinion distributions, under both homogeneous and heterogeneous susceptibility. 

\paragraph{Observed opinions under homogeneous susceptibility.}
Figure~\ref{fig:real_shift_fraction} reports the median shift toward the target as a function of the modified fraction under homogeneous susceptibility. Here, we focus on edge interventions, represented by solid curves; node-level interventions are examined in the subsequent comparison. Structural interventions produce non-negligible shifts under observed opinion distributions, although their effectiveness varies substantially across datasets. Reddit-Chitra starts close to the target, and Opinion-PR, OAES, OAES-Aware, Huber, and Sigmoid flip the median with an edge budget of $0.05|E|$. In Twitter-Chitra, the target is farther from the initial equilibrium, but several methods eventually reach it as the budget increases. Opinion-PR crosses the target at a modified fraction of $0.3$, followed by Huber and OAES-Aware at $0.4$, and OAES at $0.5$. Polblogs is more resistant to intervention, with most methods producing only limited shifts; nevertheless, Huber crosses the target at a modified fraction of $0.2$ on the same grid.  Additional targeted evaluations at finer budget resolution confirm target crossing with Huber at approximately $11\%$
edge additions on Polblogs and with Opinion-PR at $27\%$ on Twitter-Chitra. These percentages are relative to the
original edge count; the figure reports the common budget grid used for the main comparison. Thus, structural intervention remains feasible under observed opinions, but the required budget and the most effective method are dataset-dependent.

\paragraph{Interaction with heterogeneous susceptibility.} Due to space constraints, we report one illustrative heterogeneous-susceptibility case in the main paper and provide the complete results in the SI. We focus on Twitter-Chitra because it exhibits the clearest within-dataset contrast between the two centrality-based policies while keeping the network and observed opinions fixed. As shown in Figure~\ref{fig:twitter_susceptibility}, under centrality-aligned susceptibility, structural interventions produce only limited median shifts and none reaches the target. Favorable nodes make up $25.5\%$ of the network and are more central, meaning higher susceptibility assigned to them weakens the influence of these sources. Conversely, under centrality-complemented susceptibility, this relationship reverses, allowing several methods to exceed the threshold with smaller budgets.  Polblogs and Reddit-Chitra exhibit more gradual or method-dependent differences, but confirm the broader conclusion that susceptibility can alter both the initial equilibrium and the effectiveness of added edges. Complete curves for all datasets and susceptibility policies are reported in the SI.

\begin{figure*}[t]
\centering
\includegraphics[width=0.8\textwidth]{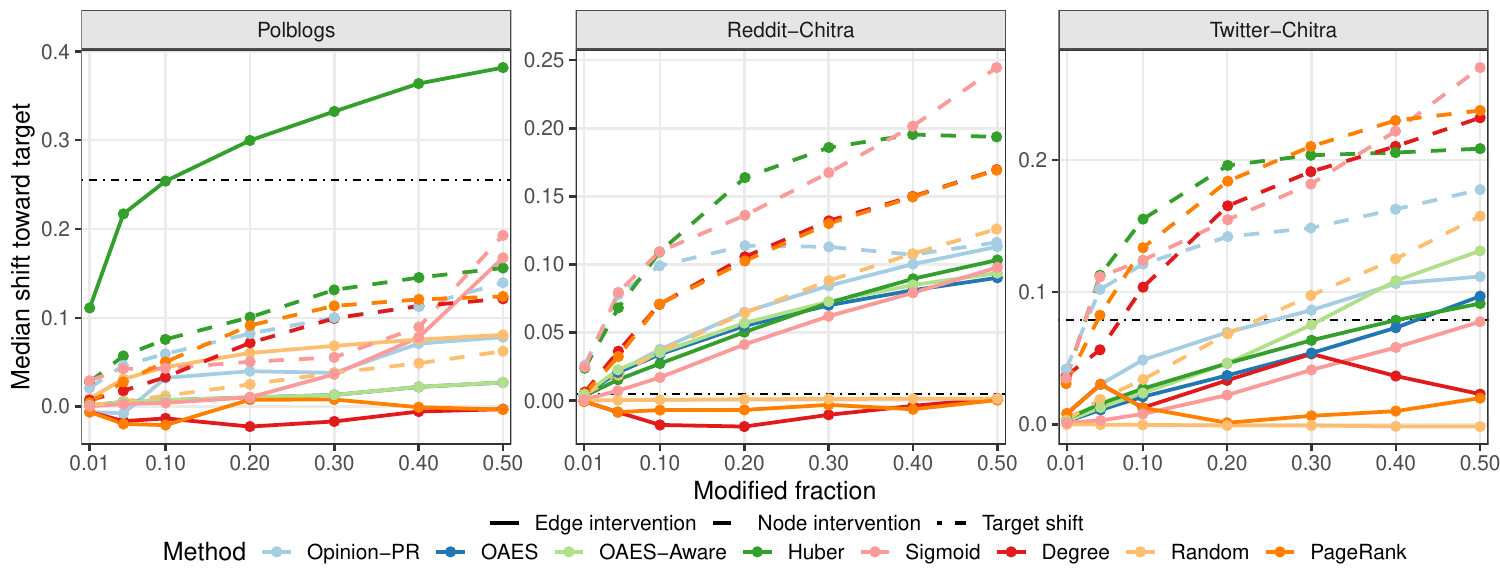}
\vspace{-5 pt}
\caption{Median shift toward the target on real-world datasets under homogeneous susceptibility. Modified fractions are \(k_{\mathrm{edge}}/|E|\) for edge interventions and \(k_{\mathrm{node}}/|V|\) for node interventions. Solid and dashed curves denote edge and node interventions, respectively, while the horizontal dot-dashed line indicates the shift required to reach the target.}
\label{fig:real_shift_fraction}
\vspace{-10pt}
\end{figure*}

\begin{figure}[h]
\centering
\includegraphics[width=\columnwidth]{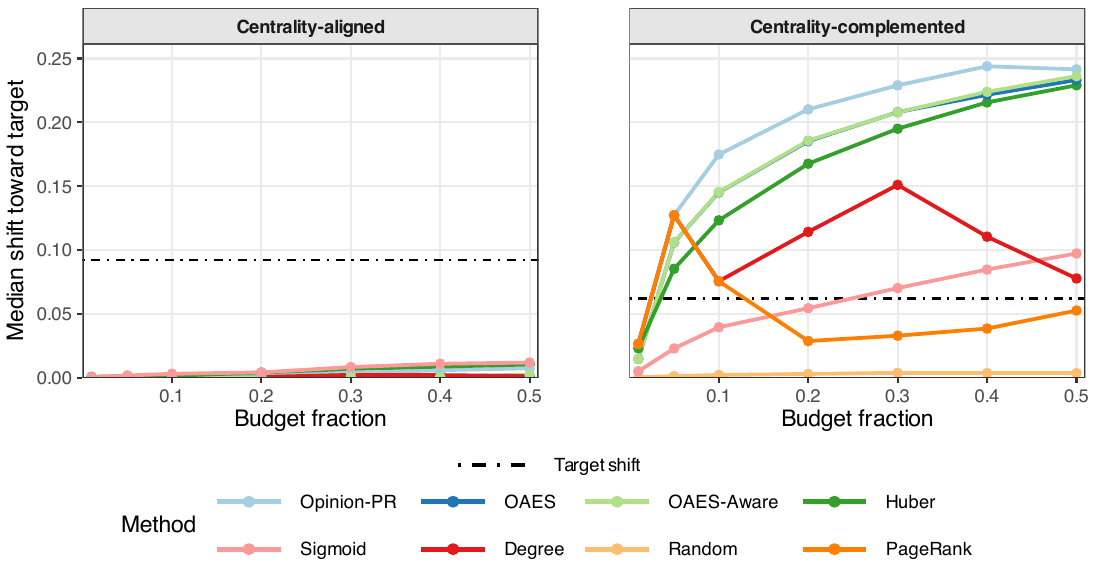}
\vspace{-15pt}
\caption{Median shift toward the target on Twitter-Chitra with centrality-aligned (left) and centrality-complemented susceptibility (right). Dot-dashed lines mark the required shift; shading shows the $95\%$ confidence interval for Random. Budget is $k_{\mathrm{edge}}/|E|$.}
\label{fig:twitter_susceptibility}
\vspace{-10pt}
\end{figure}

\subsection{RQ3: How do structural interventions compare with node-level interventions?}
\label{sec:rq3_node_edge}

We finally compare structural interventions with their node-level counterparts under homogeneous susceptibility, which provides a common baseline unaffected by the markedly different initial equilibria induced by heterogeneous susceptibility. Corresponding heterogeneous-susceptibility comparisons are reported in the SI. We first use normalized intervention fractions and then account for the different costs of modifying nodes and edges.

\subsubsection{RQ3a: Comparison under normalized budgets}\label{sec}

Figure~\ref{fig:real_shift_fraction} compares structural and node-level interventions under the same normalized modified-fraction scale. Node-level interventions often induce larger median shifts for a given modified fraction because they directly change selected users' susceptibility. Structural interventions act more indirectly by modifying exposure patterns through edge additions. Nevertheless, edge-based methods still achieve meaningful shifts and can outperform node-level interventions. In Reddit-Chitra and Twitter-Chitra, the best strategies from both intervention families cross the target with relatively small or moderate budgets. Polblogs exhibits a qualitatively different behavior: only structural interventions reach the target, while all node-level methods remain below it. This behavior is likely related to the highly assortative structure of Polblogs. Since opinions are strongly clustered within ideological communities, changing the susceptibility of a few users may have limited impact on cross-community influence, whereas adding edges can directly create new exposure paths between otherwise weakly connected regions of the network.


\subsubsection{RQ3b: Cost-calibrated comparison}
\label{sec:cost_calibrated}

Normalized fractions do not account for the different costs of manipulating users and adding edges. Table~\ref{tab:degree_weighted_cost} therefore reports the least expensive successful intervention under the degree-weighted cost model defined in the experimental setup. 
The results show that the relative advantage of node-level interventions depends strongly on the assumed cost model. When $\rho$ is small, manipulating a limited number of influential users can be less expensive than adding the edges required to reach the target. As the cost assigned to high-reach users increases, however, structural interventions become progressively more competitive. In Polblogs, as already observed in RQ3a, only structural interventions flip the median, making edge additions the only feasible strategy throughout the tested range. In Reddit-Chitra and Twitter-Chitra, node-level interventions are cheaper for $\rho=1$ and $\rho=2$. At $\rho=5$, edge additions become preferable in Twitter-Chitra, whereas node interventions remain cheaper in Reddit-Chitra. At $\rho=10$, edge additions are the least expensive strategy for all three datasets. The outcome remains unchanged for the larger values reported in the SI.

Thus, the apparent advantage of node manipulation under normalized budgets diminishes once targeting high-reach users is assigned a greater cost. Although the crossover is dataset-dependent, structural interventions become the most cost-effective option under moderate node-cost assumptions.

\begin{table}[h]
\centering
\setlength{\tabcolsep}{4pt}
\caption{Minimum intervention cost with degree-weighted node costs and homogeneous susceptibility. Each cell reports the least-cost intervention type and cost ($E$: edge addition; $N$: node-level intervention).}
\vspace{-8pt}
\label{tab:degree_weighted_cost}
\begin{tabular}{lcccc}
\toprule
Dataset & $\rho=1$ & $\rho=2$ & $\rho=5$ & $\rho=10$ \\
\midrule
Polblogs       & E (3837) & E (3837) & E (3837) & E (3837) \\
Reddit-Chitra  & N (160)  & N (320)  & N (800)  & E (897)  \\
Twitter-Chitra & N (735)  & N (1470) & E (2183) & E (2183) \\
\bottomrule
\end{tabular}
\vspace{-15pt}
\end{table}

\section{Conclusions}

We investigated whether the equilibrium median in the Friedkin-Johnsen model can be steered by modifying the interaction network rather than user-level attributes. We formulated structural median control as an edge-addition problem, proved its computational hardness, and proposed continuous optimization methods and computationally lighter opinion-aware heuristics.

Our experiments show that intervention success depends jointly on network topology, opinion placement, and susceptibility. Structural interventions are most effective when added edges create favorable influence pathways toward pivotal nodes. Community structure and heterogeneous susceptibility can substantially alter their effectiveness, while opinion-aware methods generally outperform purely topological baselines. Under homogeneous susceptibility, node-level interventions often appear stronger under normalized budgets, but structural interventions become competitive and frequently more cost-effective when targeting high-reach users is costly. These findings also provide an auditing perspective, identifying topologies and opinion distributions exposed to structural manipulation and the cost regimes in which it becomes attractive.

Future work should account for probabilistic edge adoption and optimize the expected equilibrium median. Overall, our results establish structural intervention as a promising alternative to direct node manipulation for steering collective outcomes under Friedkin--Johnsen dynamics.

\section{Acknowledgements}
Generative AI was used solely to assist with language editing and to improve the clarity and readability of the manuscript. It was not used to conduct the research, generate or analyze data, develop the methodology, or draw scientific conclusions. All AI-assisted edits were reviewed and approved by the authors.




\bibliography{references_compact}

\end{document}